\documentclass{xarticle}
\usepackage[rgb]{xcolor}
\usepackage{amsfonts,amssymb,amsmath,amsthm}
\usepackage{graphicx}
\usepackage{overpic}
\usepackage[shortlabels]{enumitem}
\usepackage[margin=10pt,font=small,skip=7pt]{caption}
\usepackage[numbers,sort]{natbib}
\usepackage{xspace}
\usepackage{stmaryrd}
\usepackage{trimclip}

\def\clap#1{\hbox to 0pt{\hss#1\hss}}

\newcommand{\R}{\mathbb{R}}
\newtheorem{thm}{Theorem}

\newtheorem{lemma}[thm]{Lemma}
\newtheorem{defn}[thm]{Definition}
\def\qedsym{\rule[0pt]{5pt}{5pt}\par\medskip}
\renewenvironment{proof}{{\noindent\bf Proof. }}{ \hfill ~\qedsym}
\newenvironment{proofnosymbol}{{\noindent\bf Proof. }}{\ignorespacesafterend}
\let\mathopfont=\mathrm

\newcommand{\ie}{{\it i.e.}}

\newcommand{\range}{\mathop{\mathopfont{range}}}
\newcommand{\abs}[1]{\lvert{#1}\rvert}
\newcommand{\eemph}[1]{\textbf{\textit{#1}}}
\def\tp{\mathsf{T}}
\newcommand{\bmat}[1]{\begin{bmatrix}#1\end{bmatrix}}
\def\one{\mathbf{1}}
\let\bl\bigl
\let\br\bigr

\newbox\vcbox
\def\vcent#1{\setbox\vcbox\hbox{#1}\raise -0.5\ht\vcbox\hbox{#1}}

\def\wu{\omega^\text{u}}
\def\prerefspace{\vspace*{-0.5ex}}
\def\precapspace{\captionsetup{margin={31pt,12pt},font=small,skip=5pt}}
\def\ugn{\lambda}
\newcommand{\bittide}{bittide\xspace}
\def\ss{{\text{ss}}}
\makeatletter
\DeclareRobustCommand{\shortto}{\mathrel{\mathpalette\short@to\relax}}
\newcommand{\short@to}[2]{\mkern2mu
  \clipbox{{.5\width} 0 0 0}{$\m@th#1\vphantom{+}{\shortrightarrow}$}}
\makeatother

\def\jtoi{{j \shortto i}}

\begin{document}

\title{On Buffer Centering for Bittide Synchronization}

\author{Sanjay Lall\footnotesymbol{1}
  \and  C\u{a}lin Ca\c{s}caval\footnotesymbol{2}
  \and  Martin Izzard\footnotesymbol{2}
  \and  Tammo Spalink\footnotesymbol{2}}

\note{Preprint}

\maketitle

\makefootnote{1}{S. Lall is with the Department of Electrical
  Engineering at Stanford University, Stanford, CA 94305, USA, and is
  a Visiting Researcher at Google.
  \texttt{lall@stanford.edu}\medskip}

\makefootnote{2}{C\u{a}lin Ca\c{s}caval, Martin Izzard, and Tammo Spalink are
  with Google.}

\begin{abstract}

  We discuss distributed reframing control of \bittide systems. In a
  \bittide system, multiple processors synchronize by monitoring
  communication over the network. The processors remain in logical
  synchrony by controlling the timing of frame transmissions. The
  protocol for doing this relies upon an underlying dynamic control
  system, where each node makes only local observations and performs
  no direct coordination with other nodes.  In this paper we develop a
  control algorithm based on the idea of reset control, which allows
  all nodes to maintain small buffer offsets while also requiring very
  little state information at each node.  We demonstrate that with
  reframing, we can achieve separate control of frequency and phase,
  allowing both the frequencies to be syntonized and the buffers to be
  moved the desired points, rather than combining their control via a
  proportional-integral controller. This offers the potential for
  simplified boot processes and failure handling.
\end{abstract}

\section{Introduction}

The Google \bittide system is designed to enable synchronous execution
at large scale without the need for a global clock. Synchronous
communication and processing offers significant benefits for
determinism, performance and utilization, and through simplification,
robustness. Synchronous execution is used successfully in real time
systems~\cite{benveniste_synchronous_2003}.

In \bittide, synchronization is decentralized, as every node in the
system adjusts its frequency based on the observed communication
exchanges with its neighbors. This mechanism is a distributed dynamic
feedback control system. The \bittide system was first proposed
in~\cite{spalink_2006}. It defines a synchronous logical clock that is
resilient to variations in physical clock frequencies.  In~\cite{bms,
  res} we discuss a model for the dynamics, the {\em Abstract Frame
  Model} (AFM) and its implication for controlling node
frequencies. This work makes use of the model and ideas developed in
that work. In this paper we focus on one critical aspect of the
design, that is control of the buffers that \bittide links use to smooth out
frequency variations.

\font\fivefont=cmss10 at 5pt
\def\doublefont{\fontencoding{T1}\fontfamily{cmss}\fontsize{11pt}{13}\selectfont}
\fboxsep 0.5pt
\def\fivept{\fivefont}
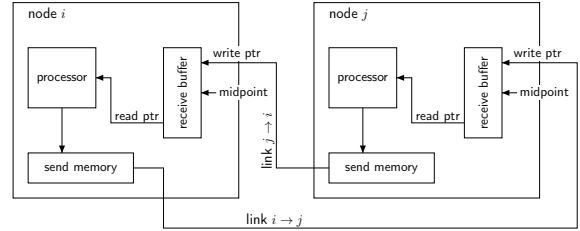
\begin{figure}
  \label{fig:buffers}
  \hbox to \linewidth{\hss\begin{overpic}[abs,width=80mm,unit=2mm]{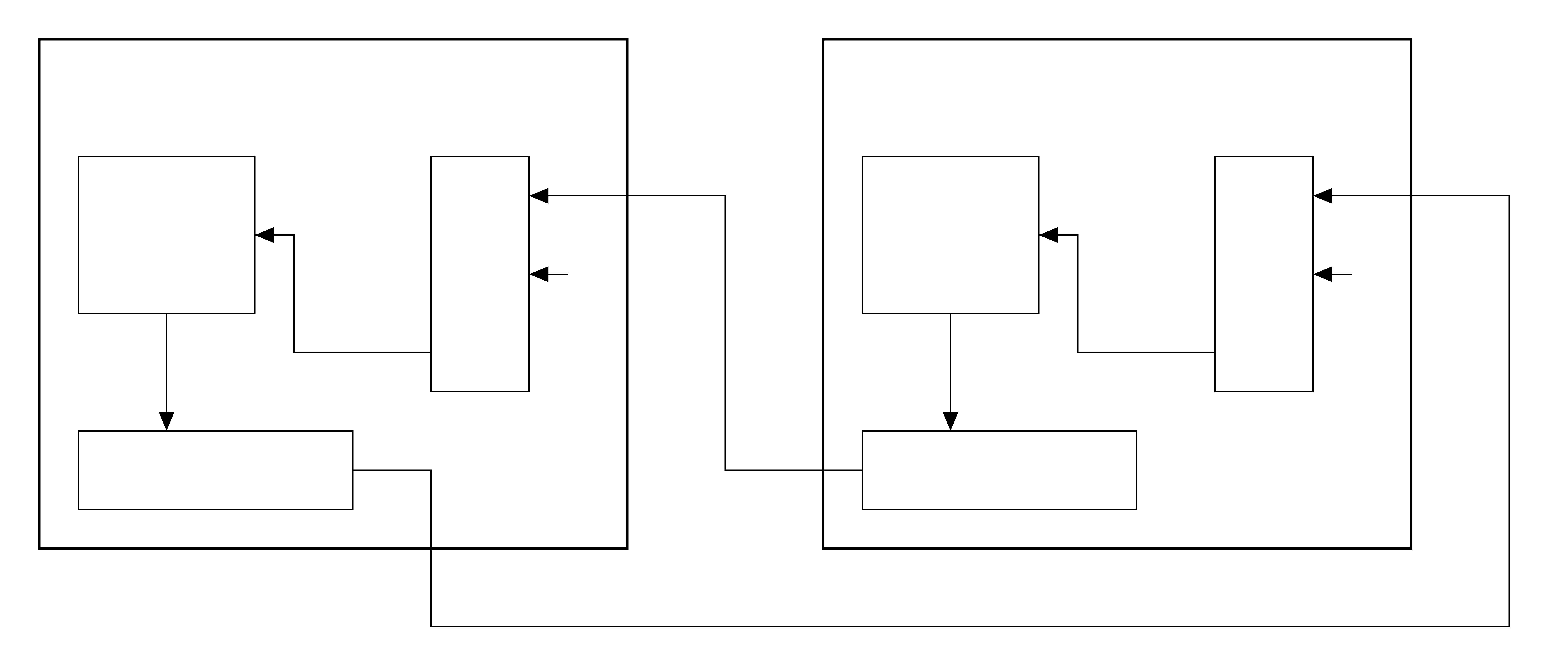}
      \put(5.5,4.8){\clap{\fivept send memory}}
      \put(4.25, 11){\vcent{\clap{\fivept processor}}}
      \put(12.25, 10){\clap{\vcent{\rotatebox{90}{\fivept receive buffer}}}}
      \put(14.6, 9.8){\colorbox{white}{\fivept midpoint}}
      \put(10.7, 8.2){\llap{\fivept read ptr}}
      \put(14.2, 12.4){\colorbox{white}{\fivept write ptr}}
      \put(2, 15){\scalebox{0.5}{{\doublefont node} $i$}}
      \put(25.5,4.8){\clap{\fivept send memory}}
      \put(24.25, 11){\vcent{\clap{\fivept processor}}}
      \put(32.25, 10){\clap{\vcent{\rotatebox{90}{\fivept receive buffer}}}}
      \put(34.6, 9.8){\colorbox{white}{\fivept midpoint}}
      \put(30.7, 8.2){\llap{\fivept read ptr}}
      \put(34.2, 12.4){\colorbox{white}{\fivept write ptr}}
      \put(22, 15){\scalebox{0.5}{{\doublefont node} $j$}}
      \put(16.5, 1.3){\scalebox{0.5}{{\doublefont link} $i \to j$}}
      \put(17.5, 4.9){\rotatebox{90}{\scalebox{0.5}{{\doublefont link} $j \to i$}}}
    \end{overpic}\hss}
  \caption{Two nodes showing the receive buffers as part of the links.}
\end{figure}

The \bittide system relies on buffering frames on the receiving side
of a link to absorb frequency oscillations in both directions. These
buffers have two primary desirable properties.  First, the control
system must avoid overflow or underflow.  Second, buffer size should be
minimized, since buffer memory is system overhead.  An initial transient
period ensures read and write pointers are at buffer midpoints,
after which their maximum needed size is determined by how much the
control system can control variation without overflow or underflow.

Processes on \bittide coordinate using ahead-of-time scheduling, and so
frame retransmissions are not allowed. Therefore, the buffers need to ensure
that no frames are lost.  Note that this does not imply that links
in the system cannot fail, simply that failures are dealt with at a
different level. For the purpose of frame delivery, buffers should
be considered reliable.

To satisfy these properties, we design a dynamic control system that
switches modes after the initial transitory period. We trade-off
convergence speed by using a proportional controller that allows
buffer occupancy to stabilize. Buffers use
virtual pointers (\ie, frame counters) to provide the correct input to the
controller. After the controller converges, the proportional
controller is replaced with a proportional-plus-offset controller to
drive the buffer occupancy to the midpoint, and retain the same
converged frequency.

\section{Prior work}

Distributed systems have long striven to achieve synchrony at some
level. This has been mostly effected by mechanisms that align local
clocks to a global master such as a UTC server. The most common are
the NTP~\cite{ntp} and PTP~\cite{ptp} standards.  Telecom systems like
SONET~\cite{sonet} and some packet networking equipment using
SyncE~\cite{syncE} apply a further level of effort to aid synchrony,
specifically, they work to syntonize local time-reference oscillators
in hierarchical manner. Syntony here is used to mean that all
reference oscillators in the distributed system maintain the same
frequency on average over time. There are well-known methods to
achieve this, mostly by clock extraction from upstream communication
links. Having syntonized oscillators eases the task of aligning local
clocks to a master.

The \bittide system also uses syntonization, although there is no
inherent need for hierarchy, and instead \bittide adds precise phase
control at the lowest level.  This phase control is not quite phase
alignment across the distributed system because there will be temporal
{\em wobble}. The \bittide system uses small elastic buffers to absorb
this wobble.  This opens the possibility of very precise coordination
across a \bittide distributed system, coordination equivalent to what
is possible in a traditional synchronous system.

A model for the dynamics of a bittide system was developed
in~\cite{bms, res}.  In this paper we focus on a simplified linear
version of this model.  Many similar linear consensus models have been
extensively studied; see~\cite{dorfler} for a survey focused
on synchronization, and see~\cite{moreau2004} for a discussion
of stability analysis.

This paper develops an approach to control of \bittide systems we call
\emph{reframing} of a \bittide system. This is related to the idea of
\emph{reset} which has a long history in control. Before the
development of integral control, in order to ensure that a system
achieved a small steady-state error, the offset (or \emph{reset})
parameter in a proportional controller was adjusted manually.
Integral control was developed as a way of automatically performing
such resets~\cite{astrom}. Here, we use a variant of the reset idea
for a \emph{distributed} system, to simultaneously control \emph{two}
quantities per node, frequency and buffer occupancy. Because the
system being controlled is a computer network, it has sufficiently
ideal properties (such as conservation of frames~\cite{bms}) that we
can perform a reset exactly once, at bootup. The idea of reset
is attributed~\cite{auslander} to Mason in the~1930s.

\section{Notation and preliminaries}

We represent the \bittide topology as a directed graph with $n$ nodes
and $m$ edges. Define the source incidence matrix $S \in\R^{n \times
  m}$ by
\[
S_{ie} = \begin{cases}
  1 & \text{if  node $i$ is the source of edge $e$} \\
  0 &\text{otherwise}
\end{cases}
\]
and the destination incidence matrix $D \in\R^{n \times m}$ by
\[
D_{ie} = \begin{cases}
  1 & \text{if node $i$ is the destination of edge $e$} \\
  0 &\text{otherwise}
\end{cases}
\]
The usual incidence matrix of the graph is then $B = S - D$.  Let
$\one$ be the vector of all ones, then $B^\tp \one = 0$.

A directed graph is called \eemph{strongly connected} or
\eemph{irreducible} if for every $i,j$ there exists directed paths $i
\to j$ and $j \to i$. Suppose $A\in\R^{n\times n}$ is a nonnegative
matrix such that $A_{ij} > 0$ if there is an edge $i \to j$ and
$A_{ij}=0$ otherwise. The matrix $A$ is called irreducible if the
corresponding graph is irreducible. Note that this does not depend on
the diagonal elements of $A$.

A matrix $Q\in\R^{n\times n}$ is called \eemph{Metzler} if $Q_{ij}\geq
0$ for all $i \neq j$. A Metzler matrix $Q$ is called a \eemph{rate
  matrix} if its rows sum to zero. If $Q$ is Metzler and irreducible,
then there is a real eigenvalue $\lambda_\text{metzler}$, with
positive left and right eigenvectors.  All other eigenvalues $\lambda$
satisfy $\Re(\lambda) < \lambda_\text{metzler}$.

A Metzler matrix $Q$ has a nonnegative matrix exponential.  To see
this, let $s >0$ be such that $sI + Q \geq 0$. Then $e^{Q} =
e^{-sI}e^{sI + Q}$ and both terms on the RHS are elementwise
nonnegative.  For the special case of a rate matrix $Q$,
since $Q \one= 0$ we have directly $e^Q \one = \one$ and so $e^Q$ is a stochastic
matrix.

\section{The \bittide control system}

A detailed model of the \bittide system, called the \emph{abstract
frame model} is developed in~\cite{bms}. A simplified
finite-dimensional linear time-invariant model was presented
in~\cite{res}. In both papers, we considered the special case of a system
where all links are bidirectional. We now
update the model to include unidirectional links. We focus on the case
where the controller is continuous-time, latencies are small, and
measurements are unquantized. The effectiveness of this approximation
was investigated in~\cite{bms}, so we do not dwell here on this issue.

The simplified fundamental dynamics of the \bittide system are as follows.
\begin{align}
  \label{eqn::phase}
  \dot\theta_i(t) &= \omega_i(t) \\
  \label{eqn:occ}
  \beta_\jtoi(t) &= \theta_j(t) - \theta_i(t) + \ugn_\jtoi\\
  \omega_i(t) &= \wu_i + c_i(t)
\end{align}
Here $i,j\in 1,\dots,n$ index nodes in the graph, and $j\to i$ refers
to an edge from $j$ to $i$. The variable $\theta_i$ is the \emph{clock
phase} at node $i$, whose time-derivative $\omega_i(t)$ is the clock
frequency. The clock frequency is the sum of two terms, the first is
the constant $\wu_i$, which is the \emph{uncontrolled frequency} of
the clock. It is unknown, and not available to the \bittide control
system. The second term is $c_i$, the \emph{frequency correction},
which is the control input; it is chosen by the controller at
node~$i$. Equation~\eqref{eqn:occ} gives $\beta_\jtoi$, the occupancy
of the elastic buffer at node $i$ associated with the edge $j \to
i$. The quantity $\ugn_\jtoi$ is a constant associated with the link.

Control for the \bittide system is inherently distributed and as such,
does not have access to global information. At each
node $i$, the controller measures all of the occupancies $\beta_\jtoi$
for all of the incoming links. It cannot observe the occupancies at
other nodes, nor can it observe~$t$ or $\dot\theta_i$.  Using this limited
information, it chooses the frequency correction $c_i$. Dynamic
controllers cannot be implemented exactly, since the controller does
not know the time $t$, and so, for example, the controller cannot
exactly perform any integration or differentiation. The only clock
information available at node $i$ is the clock phase~$\theta_i$, and
in general this varies from one node to the next. This means that we
are desirous of implementing a purely static controller, such as a
proportional (plus-offset) controller. One form, which has been studied
in~\cite{res}, is
\begin{equation}
  \label{eqn:pcontroller}
  c_i(t) = k \sum_{j \mid \jtoi} (\beta_\jtoi - \beta_i^\text{off}) + q_i(t)
\end{equation}
The buffer occupancy $\beta_\jtoi$ is measured relative to an offset
$\beta^\text{off}_i$, corresponding to the desired equilibrium buffer
occupancy. The difference $\beta_\jtoi -\beta^\text{off}_i$ is called
the \emph{relative} buffer occupancy.  The controller chooses the
correction to be proportional to the sum of the relative buffer
occupancies at the node, plus a constant frequency offset~$q_i$. The
controller parameter $k$ could in principle depend on the node $i$,
but for simplicity and scalability we do not consider that case. As we
discuss in this paper, the frequency offset $q_i(t)$ may vary with
both time and node.

\begin{figure}[ht!]
  \centerline{\scalebox{0.75}{\begin{overpic}[abs,width=72mm,unit=0.8mm]{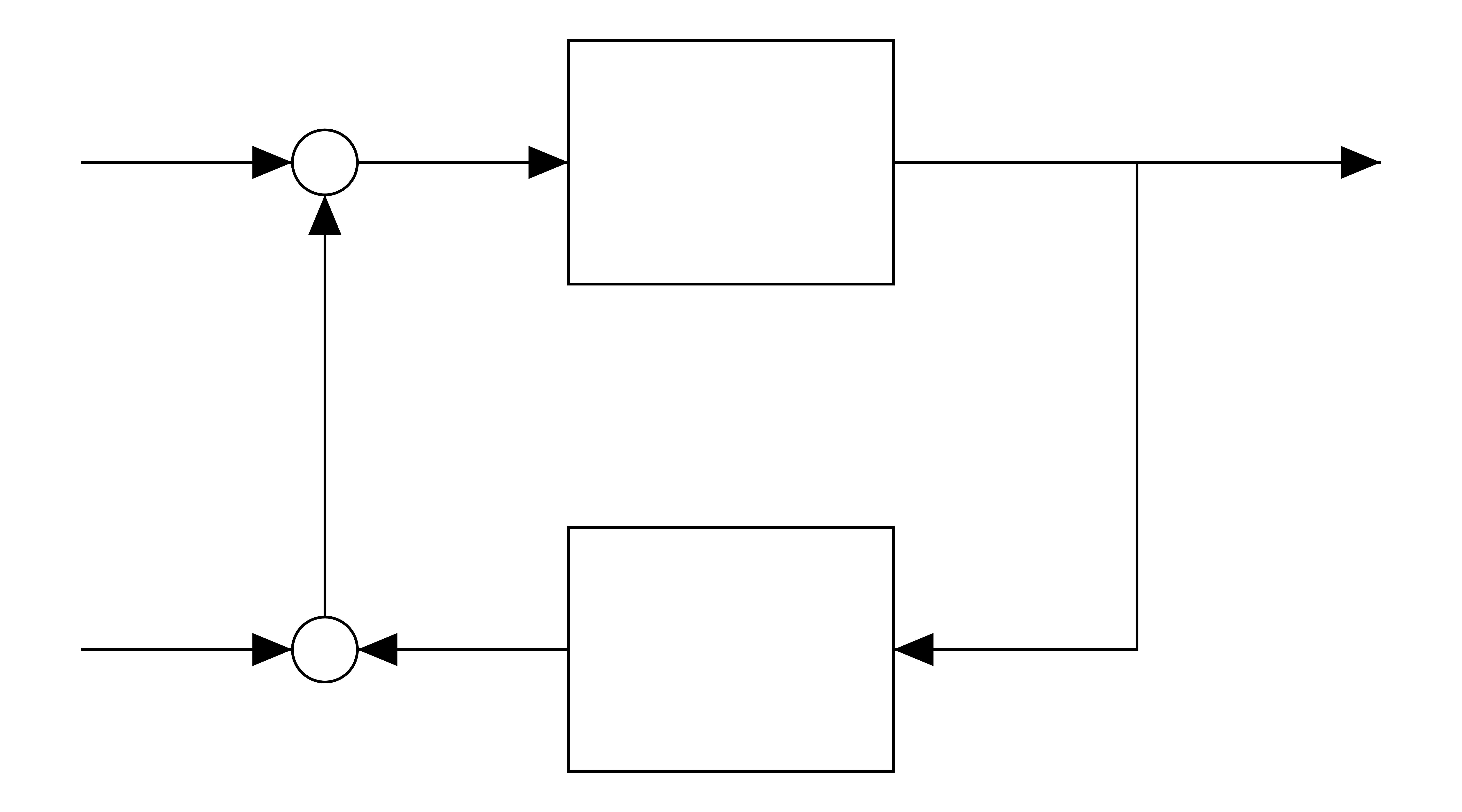}
      \put(87,39){\small$\beta$}
      \put(-1,39){\small$\wu$}
      \put(-1,9){\small$q$}
      \put(21,25){\vcent{\hbox{$c$}}}
      \put(28,41.5){\clap{$w$}}
      \put(45,10){\clap{\vcent{\hbox{$K$}}}}
      \put(45,40){\clap{\vcent{\hbox{$G$}}}}
    \end{overpic}}}
  \precapspace
  \caption{Feedback block diagram}
  \label{fig:blockdiag}
\end{figure}
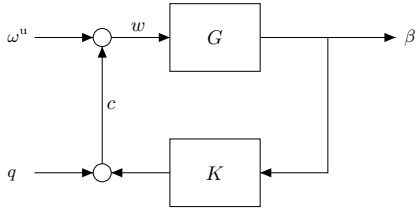

The controller interconnection is illustrated in
Figure~\ref{fig:blockdiag}.  The system model $G$ maps frequency
$\omega$ to buffer occupancy $\beta$, and the controller $K$ maps
$\beta$ to correction~$c$ minus the offset $q$.

\subsection{Model}

The model for the closed-loop system is described in vector form
as follows:
\begin{equation}
  \label{eqn:mod1}
  \begin{aligned}
    \dot\theta &= \wu + c \\
    \beta &= B^\tp \theta + \lambda \\
    c &= kD (\beta - \beta^\text{off}) + q
  \end{aligned}
\end{equation}
Here $\beta,\lambda \in\R^m$ and $\theta, c, q\in\R^n$. It's
convenient to write this as
\begin{equation}
  \label{eqn:bit}
  \begin{aligned}
    \dot\theta &= A\theta + \wu + q + r \\
    \beta &= B^\tp \theta + \ugn \\
    c &= A \theta + q + r
  \end{aligned}
\end{equation}
where
\begin{equation}
  \label{eqn:Ar}
  A = kDB^\tp \qquad r = kD(\ugn - \beta^\text{off})
\end{equation}
Note that the matrix $A$ is not Hurwitz, and so $\theta$ does not
converge.

When the system is booted up, the
offsets are chosen to be \emph{feasible}, that is we set
$\beta^\text{off} = \beta(t^0)$ at some time $t^0$. This has the following consequence.

\begin{lemma}
  \label{lem:feasible}
  Suppose $\beta^\text{off}$ is feasible, that is, there exists $t^0$ such that
  \[
  \beta^\text{off}(t^0) = B^\tp \theta(t^0) + \lambda
  \]
  Let $r$ be given by equation~\eqref{eqn:Ar}. Then $r\in\range(A)$.
\end{lemma}
\begin{proof}
  This holds because
  $  r = kD(\lambda - \beta^\text{off})$ and so $r = -kDB^\tp \theta(t^0) = - A\theta(t^0)$.
\end{proof}
We assume the graph is irreducible. Then the matrix~$A$ is a
scaled directed Laplacian matrix for the graph. It is an irreducible
rate matrix.  The following result is standard.

\begin{lemma}
  \label{lem:W}
  Suppose $A$ is an irreducible rate matrix, and let $z>0$ be it's
  Metzler eigenvector, normalized so that $\one^\tp z = 1$.  Then
  \[
  \lim_{t \to \infty} e^{At} = \one z^\tp
  \]
\end{lemma}
\begin{proof}
  Since $A$ is irreducible, the Metzler eigenvalue, which is zero, has
  multiplicity one. Let the eigendecomposition of $A$ be $  AT = TD$.
  Then we have
  \[
  D = \bmat{0 & 0 \\ 0 & \Lambda}
  \quad
  T = \bmat{\one & T_2}
  \quad
  T^{-1} = \bmat{z^\tp \\ V_2^\tp}
  \]
  for appropriate matrices $T_2$, $V_2$, and $\Lambda$.
  The other eigenvalues of $A$ all have negative real part.
  Since $T^{-1}T=I$ we have $ \one^\tp z = 1$.
  Then
  \[
  e^{At} = T
  \bmat{1 & 0 \\ 0 &   e^{\Lambda t}} T^{-1}
  \]
  Taking the limit gives the result.
\end{proof}

We denote by $z>0$ the Metzler left-eigenvector of $A$, normalized so
that $\one^\tp z= 1$, and let $W = \one z^\tp$.  The matrix $W$ is
called the spectral projector of the Metzler eigenvalue. It satisfies
$W^2 = W$ and in addition $WA = AW = 0$.

\section{Approach}

There are many important practical requirements that the controller
must meet, which are discussed in~\cite{bms}.  In this
paper, we focus on two critical requirements. The first is that the
controller must ensure that the frequency of all nodes converges
to the same value,
that is, for some $\bar\omega>0$, we have
\[
\lim_{t\to\infty} \dot\theta_i(t) = \bar\omega \quad
\text{for all }i
\]
The second requirement is that, after some initial startup time $T$,
the buffer occupancies remain close to the offset, that is,
$\abs{\beta_i(t) - \beta^\text{off}_i}$ should be small for all $i$
and all $t > T$.  Both of these requirements are specifications of
allowed steady-state behavior.

One approach to control this system is to use an
approximate proportional-integral (PI) control, as discussed
in~\cite{res,bms}.  Even in situations where the effects of the
approximation are small, there are potential disadvantages to the PI
controller.  One is that the PI controller contains the integral
state, which the controller may need to set carefully when a node
starts up and when neighboring nodes fail. In the distributed setting
of \bittide, one of the design tenets is to avoid in-band signaling for
controlling synchronization, and thus we have no mechanism to exchange
such information.  Appropriate choice of the integral
gain may be affected by the underlying network topology, link rates,
and link latencies.

Another alternative method to proportional-integral control is simply using a
very large gain $k$. Very large gains have negative consequences for feedback
system behavior in several well-known ways. In particular for bittide, this
would adversely affect delay robustness and response to quantization noise,
both of which are important in this setting.  We therefore develop an
alternative approach in this paper.

\paragraph{Reframing control.}
We give here for convenience a brief summary of the technical approach,
which is detailed precisely in Section~\ref{sec:main}.
Our approach is to make use of the offset term $q$ in the
proportional-plus-offset controller. At each node, the local
controller sets $q = 0$ initially. The system frequency $\omega$ will
then converge, so that all nodes have the same frequency, a weighted
average of $\wu$. This is shown in Lemma~\ref{lem:freqconv}.

The buffer occupancy also converges, as shown in
Lemma~\ref{lem:betaconv}.  Typically the buffer occupancy will not
converge to the buffer offset, because a nonzero correction $c$ given
by equation~\eqref{eqn:css} is necessary to maintain frequency
equilibrium, and with $q=0$ the controller~\eqref{eqn:pcontroller} can
only achieve this with a nonzero relative occupancy~$\beta - \beta^\text{off}$.

After this initial period of convergence, the correction has converged
to a steady-state value $c^\text{ss}$. The controller now performs a
\emph{reframing}; it sets the offset $q$ in the controller to
$c^\text{ss}$. The control signal emitted by the controller is now not
an equilibrium solution, and so the system will need to reconverge to
a new equilibrium. The key point here is that the new equilibrium is
at the same \emph{frequency}, but a different \emph{buffer occupancy}.
We show in Lemma~\ref{lem:cconv} that, after the reframing, the
frequency converges to the same weighted equilibrium value as before.
Furthermore, Lemma~\ref{lem:betaconv2} shows that after this second
phase, the buffer occupancy $\beta(t)$ converges to the midpoint
$\beta^\text{off}$.

By using this two-phase approach, we can therefore achieve \emph{both}
controller requirements. The steady-state frequency is exactly the
same as that achieved by the proportional controller, and the final
buffer occupancy is at the desired offset point.

\section{Main results}
\label{sec:main}

We first analyze simple convergence. The clock phase $\theta$ does not
converge, so there is no steady-state value for it.  However, the
frequency $\omega(t) = \dot\theta(t)$ does converge.
We will make use of the following simple property.
\begin{lemma}
  \label{lem:conv}
  Let $A$ be an irreducible rate matrix, and $\dot \theta(t) = A
  \theta(t) + v$.  Let $W$ be the spectral projector corresponding to
  the Metzler eigenvalue. Then
  \[
  \lim_{t \to \infty} A \theta(t) =  (W-I) v
  \]
  for any initial conditions $\theta(0)$.
\end{lemma}
\begin{proof}
  We have
  \[
  \theta(t) = \int_0^t e^{A(t-s)}v \, ds  + e^{At} \theta(0)
  \]
  and hence
  \[
  A\theta(t) = (e^{At}-I) v + A e^{At} \theta(0)
  \]
  Taking the limit gives the desired result.
\end{proof}
\noindent
Define for convenience the function $F:\R^n \to \R^n$ by
\[
F(q) =  (W-I) \wu + W(q+r)
\]
This function maps the controller offset $q$ to the steady-state
correction, as follows.
\begin{lemma}
  \label{lem:freqconv}
  Consider the dynamics of~\eqref{eqn:bit}.
  For any initial $\theta(0)$ and any $q\in\R^n$ we have
  \begin{equation}
    \label{eqn:css}
    \lim_{t \to \infty} c(t) = F(q)
  \end{equation}
\end{lemma}
\begin{proof}
  This is a direct consequence of Lemma~\ref{lem:conv}.
\end{proof}
The frequency of the system is $\omega(t) =  c(t) + \wu$ which
gives steady-state frequency
\[
\omega^\text{ss} = \lim_{t \to \infty} \omega(t) = W(q + r + \wu)
\]
Hence, before reframing, the frequency converges. The frequency that
the system converges to meets the first performance requirement of
\bittide, since we have $W=\one z^\tp$ and therefore we must have all
components of $\omega^\ss$ equal. Since $\one^\tp z = 1$ the
steady-state frequency is a convex combination of the entries of $q+r+\wu$.  From
Lemma~\ref{lem:feasible}, we have $r\in\range(A)$. Then
\[
\omega^\text{ss} = W(q + \wu)
\]
since $WA=0$. Then we can see that a pure proportional controller
(\ie, $q=0$) results in all nodes of the \bittide system converging to a
steady-state frequency equal to the weighted average $z^\tp \wu$, that is
\begin{equation}
  \label{eqn:wbefore}
  \omega^\text{ss} = W \wu = \one z^\tp \wu
\end{equation}

The steady-state frequency is affected by~$q$, and so even though this
frequency meets the \bittide requirement, after reframing we will
adjust $q$, and so it appears in principle possible that the frequency
of the system could change. We will see below that, although at the
time of reframing the frequency of the system does change, it
nonetheless returns and converges to the same frequency.  An example
of this behavior is given in Figure~\ref{fig:freq}.  The first
\bittide requirement is satisfied; for any positive gain~$k$ the nodes
all converge to the same steady-state frequency.  We now turn
attention to the second requirement, ensuring that eventually $\beta$
stays close to the offset $\beta^\text{off}$. We first show that,
before reframing, the buffer occupancy converges.

\begin{lemma}
  \label{lem:betaconv}
  Consider the dynamics of~\eqref{eqn:bit}.
  For any initial $\theta(0)$ and any $q\in\R^n$, the buffer occupancy
  $\beta(t)$ converges as $t \to \infty$.
\end{lemma}
\begin{proofnosymbol}
  Using the same notation as in the proof of Lemma~\ref{lem:W}, we have
  \[
  \int_0^t e^{A(t-s)}\, ds = T
  \bmat{t & 0 \\ 0  & \Lambda^{-1} (e^{\Lambda t} - I)} T^{-1}
  \]
  Then
  \begin{align*}
    \beta(t) &= \lambda + B^\tp \theta(t)  \\
    &= \lambda + B^\tp T
    \bmat{0 & 0 \\ 0  & \Lambda^{-1} (e^{\Lambda t} - I)} T^{-1} (\wu + q + r)
    \\
    & \qquad\qquad
    + B^\tp T \bmat{ 0 & 0 \\ 0 & e^{\Lambda t}} T^{-1} \theta(t)
  \end{align*}
  where we have used the fact that the first column of $T$ is~$\one$
  and $B^\tp \one = 0$. Then
  \[
  \lim_{t \to \infty} \beta(t)
  = \lambda - T_2 \Lambda^{-1} V_2^\tp (\wu + q + r)
  \tag*{\qedsym}
  \]
\end{proofnosymbol}

\begin{figure}[ht!]
  \centerline{\begin{overpic}[width=0.25\linewidth]{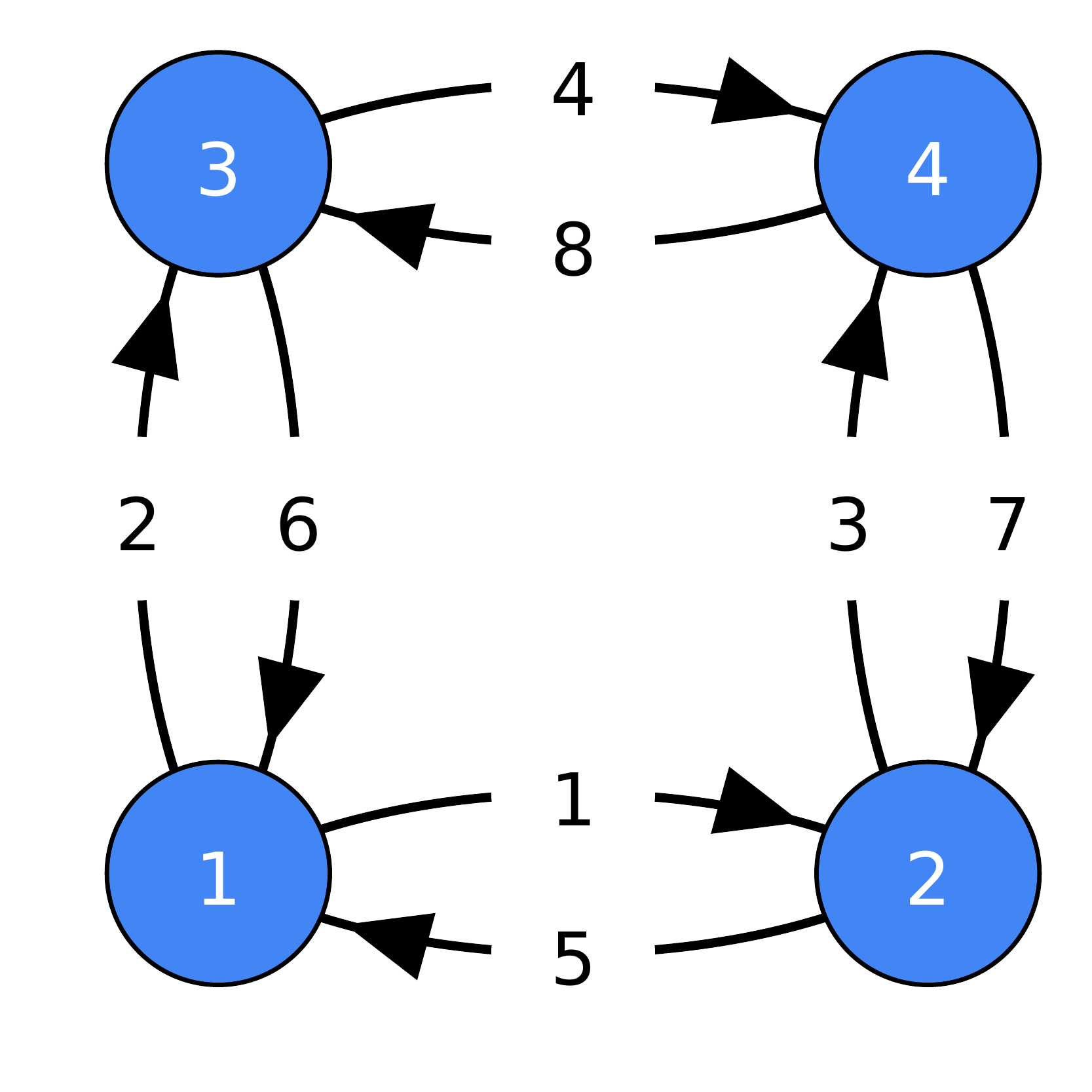}\end{overpic}}
  \precapspace
  \captionsetup{margin=30pt}
  \caption{Graph used to generate the simulations of Figures~\ref{fig:freq}
    and~\ref{fig:mocc}.}
  \label{fig:graph}
\end{figure}

\begin{figure}[ht!]
  \centerline{\begin{overpic}[width=0.98\linewidth]{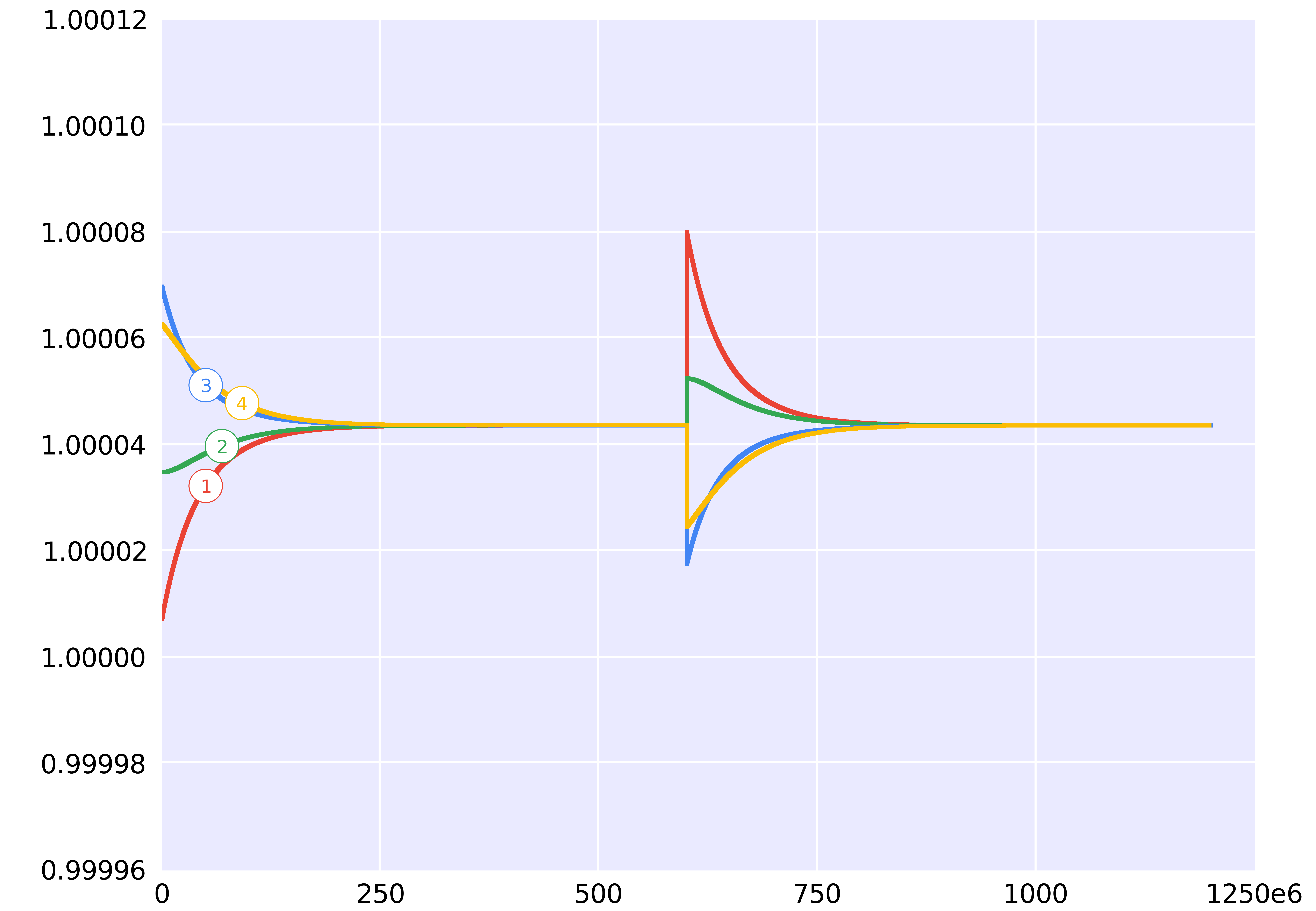}
      \put(13,64){\small$\omega$}
      \put(93,4){\small$t$}
    \end{overpic}}
  \precapspace
  \caption{Frequency behavior as a function of time. The reframing occurs
    at about $t=600\text{ns}$. Simulations were performed using the \texttt{Callisto}~\cite{callisto}
    simulator with a detailed frame-accurate model of \bittide. Traces on the graph are labeled
  with the corresponding node number.}
  \label{fig:freq}
\end{figure}

\begin{figure}[ht!]
  \centerline{\begin{overpic}[width=0.98\linewidth]{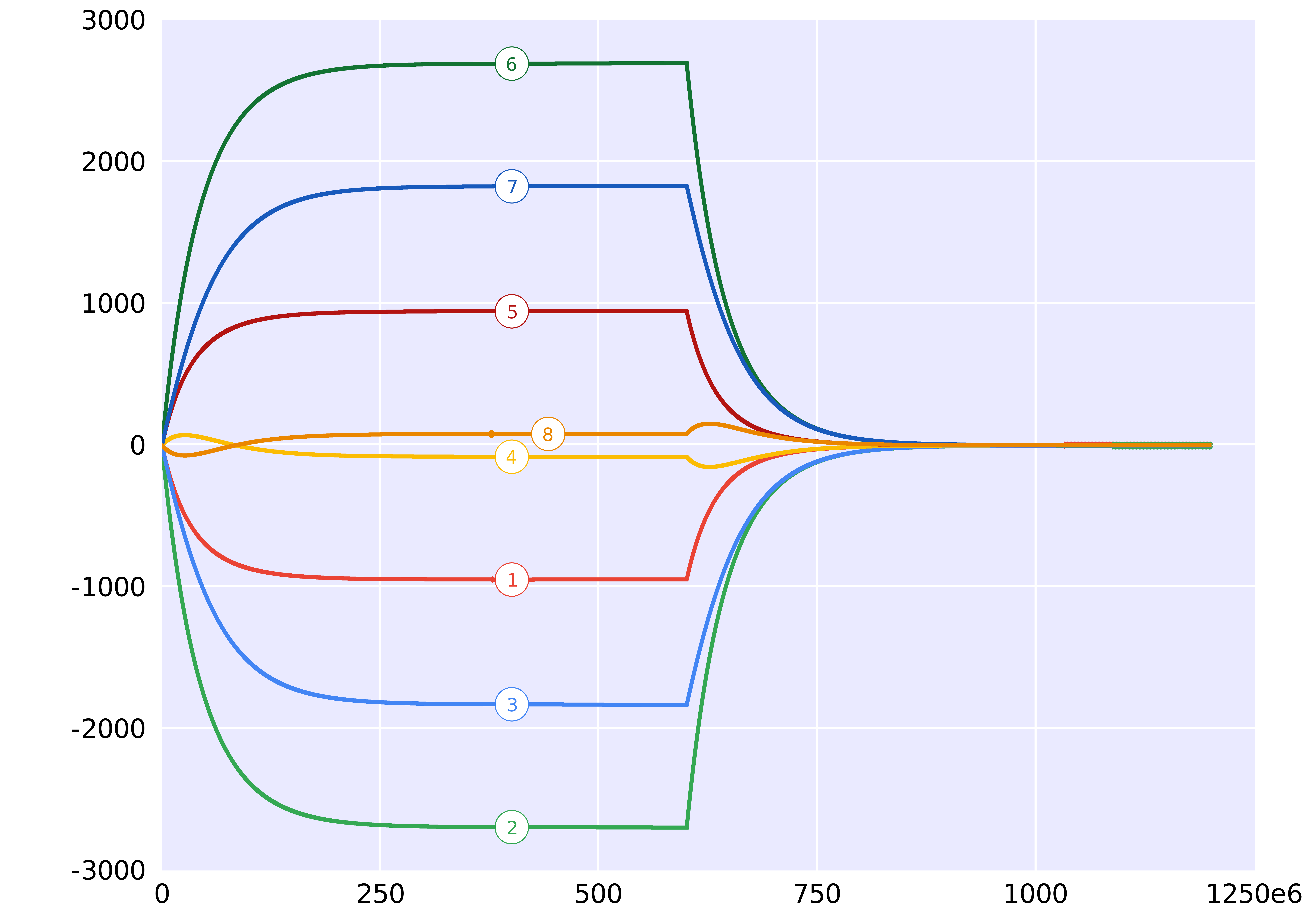}
      \put(13,64){\small$\beta - \beta^\text{off}$}
      \put(93,4){\small$t$}
    \end{overpic}}
  \precapspace
  \caption{Relative buffer occupancies $\beta - \beta^\text{off}$ for the same
  simulation as in Figure~\ref{fig:freq}. Traces on the graph are labeled
  with the corresponding edge number.}
  \label{fig:mocc}
\end{figure}

We now turn to the reframing. The controller runs a proportional
controller for some amount of time $T_1$, long enough to ensure that,
in practice, the frequency and the buffer offsets have
converged. After that time, the controller changes to using a non-zero
offset, which is simply equal to the converged value of the
correction. This controller is stated formally below.

\begin{defn}
  We define the \eemph{reframing controller} as follows.  For some
  $T_1>0$, let the correction be
  \[
  c(t) = \begin{cases}
    kD (\beta(t) - \beta^\text{off})  & \text{for } t \leq T_1 \\
    kD (\beta(t) - \beta^\text{off}) + kD (\beta(T_1) - \beta^\text{off})
     & \text{otherwise}
  \end{cases}
  \]
\end{defn}

Now we show the desired frequency convergence property. After
reframing, the controller converges to the same frequency as that
before reframing, in equation~\eqref{eqn:wbefore}.
\begin{lemma}
  \label{lem:cconv}
  Suppose $\beta^\text{off}$ is feasible.
  Using the reframing controller, as $T_1 \to \infty$
  and $t \to  \infty$, the frequency converges
  \[
  \omega(t) \to W \omega^u
  \]
\end{lemma}
\begin{proof}
  Since we are considering both $T_1$ and $t$ large, we can evaluate
  convergence in two phases. In the first phase we have $q=0$ and so
  according to Lemma~\ref{lem:freqconv} we
  have $c(T_1) \to  F(0)$. The reframing controller is
  \[
  c(t) =  kD (\beta(t) - \beta^\text{off}) + c(T_1) \qquad\text{for } t > T_1
  \]
  We can therefore use Lemma~\ref{lem:freqconv} again, with $q=c(T_1)$, to give
  \begin{align*}
    \smash{\lim_{t \to \infty}} c(t)  &= F(F(0))\\
    &= F\bl(
    (W-I)\wu + Wr
    \br)\\
    &= (W-I) \wu + W( r + (W-I) \wu + Wr)
    \\
    &= (W-I)\wu + 2Wr \\
    &= (W-I)\wu
  \end{align*}
  where the last line holds since $\beta^\text{off}$ is feasible. Then since
  $\omega(t) = c(t) + \wu$ we have
  \[
  \lim_{t \to \infty} \omega(t)  = W \wu
  \]
  as desired.
\end{proof}

Finally, we turn to the critical requirement for \bittide, that the
buffer occupancies be kept close to the offset. The following result
shows that, after the reframing, buffer occupancies return to the
midpoint. This is illustrated by the simulation in
Figure~\ref{fig:mocc}.

\begin{lemma}
  \label{lem:betaconv2}
  Suppose $\beta^\text{off}$ is feasible.
  Using the reframing controller, as $T_1 \to \infty$
  and $t \to  \infty$, the buffer occupancy  converges
  \[
  \beta(t) \to \beta^\text{off}
  \]
\end{lemma}
\begin{proof}
  Denote $\beta^\text{ss} = \lim_{t \to \infty} \beta(t)$.
  With the reframing controller, we have
  \[
  \lim_{t\to\infty}
  c(t) =
  kD(\beta^\text{ss} - \beta^\text{off}) + c(T_1)
  \]
  The proof of Lemma~\ref{lem:cconv} shows that
$
  \lim_{t \to \infty} c(t)  = c(T_1)
 $
  and therefore
  \begin{equation}
    \label{eqn:bfin}
    D(\beta^\text{ss} - \beta^\text{off})  = 0
  \end{equation}
  Now, since $r=0$, we have $\beta^\text{off} = B^\tp \theta(t^0) + \lambda$, and so
  \[
  \beta(t) - \beta^\text{off} = B^\tp (\theta(t) - \theta(t^0))
  \]
  hence for all $t$ we have $\beta(t) - \beta^\text{off} \in
  \range(B^\tp)$. Hence
  $\beta^\text{ss} - \beta^\text{off} = B^\tp x$ for some $x$, and
  so~\eqref{eqn:bfin} means that $DB^\tp x = 0$. Since $DB^\tp$ is an irreducible
  rate matrix, this means $x=\gamma\one$ for some $\gamma$, and hence $B^\tp x = 0$,
  from which we have $  \beta^\text{ss} - \beta^\text{off} = 0$ as desired.
\end{proof}

Lemmas~\ref{lem:cconv} and~\ref{lem:betaconv2} show that the reframing
methodology achieves the requirements. In practice, this works
because, before reframing, the buffers can overflow. This occurs
during system startup, when the network frames do not contain any
data, and so the system does not need to store the actual frames, and
instead can use counters or pointers to keep track of how many frames
\emph{would} be in the buffer. After the reframing, and subsequent
convergence, the \bittide system can begin executing code, at which
time it is essential that frames not be dropped. At this point, the
buffer occupancies have returned to a stable equilibrium at the
midpoint.

\section{Conclusions}

We proved that we can satisfy the requirements for controlling a \bittide
system and satisfy the desirable properties for buffer occupancy by developing
a dynamic control system that resets after convergence. The initial
proportional controller drives frequency convergence, and the proportional plus
offset controller ensures that buffer occupancies are driven toward the
desirable midpoint.

\section{Acknowledgments}

We thank Robert O'Callahan, Pouya Dormiani, Chase Hensel, and Chris
Pearce for all of their work on this project, and for much
collaboration and helpful discussion.

\prerefspace
\bibliographystyle{abbrv}

\end{document}